\documentclass[american,aps,pra,reprint,tightenlines,superscriptaddress,twocolumn,twoside]{revtex4-1}

\usepackage[utf8]{inputenc}
\usepackage{amsmath, amsthm, amssymb}
\usepackage{amsthm}
\usepackage{amsfonts}
\usepackage{tikz}
\usepackage{braket}
\usepackage{graphicx}
\usepackage[english]{babel}
\usepackage{comment}
\usepackage{caption}
\usepackage{subfigure}
\usepackage{graphicx}               
\usepackage[colorlinks=true,linkcolor=blue,urlcolor=blue,citecolor=blue]{hyperref}
\newtheorem{prop}{Proposition}

\definecolor{cadetblue}{rgb}{0.37, 0.62, 0.63}

 \newtheorem*{remark}{Remark}

\begin{document}
\title{Detecting entanglement harnessing Lindblad structure}

\author{Vaibhav Chimalgi} 
\affiliation{Center for Security Theory and Algorithmic Research,
International Institute of Information Technology, Gachibowli, Hyderabad 500 032, India.}

\author{Bihalan Bhattacharya} 
\email{bihalan@gmail.com}
\affiliation{Department of Mathematical Sciences, Indian Institute of Science Education and Research Berhampur, Transit Campus,
Government ITI, Berhampur 760010, Odisha, India.}

\author{Suchetana Goswami} 
\email{suchetana.goswami@gmail.com}
\affiliation{Harish-Chandra Research Institute, A CI of Homi Bhabha National Institute, Chhatnag Road, Jhunsi, Allahabad 211 019, India}

\author{Samyadeb Bhattacharya} 
\email{samyadeb.b@iiit.ac.in}
\affiliation{Center for Security Theory and Algorithmic Research,
International Institute of Information Technology, Gachibowli, Hyderabad 500 032, India.}


\begin{abstract}
    The problem of entanglement detection is a long standing problem in quantum information theory. One of the primary procedures of detecting entanglement is to find the suitable positive but non-completely positive maps. Here we try to give a generic prescription to construct a positive map that can be useful for such scenarios. We study a class of positive maps arising from Lindblad structures. We show that  two famous positive maps \textit{viz.} transposition and Choi map can be obtained as a special case of a class of positive maps having Lindblad structure. Generalizing the transposition map to a one parameter family we have used it to detect  genuine multipartite entanglement. Finally being motivated by the negativity of entanglement, we have defined a similar measure for genuine multipartite entanglement.   
\end{abstract}

\maketitle


\section{Introduction}

The core structure of quantum information processing is predominantly governed by presence of quantum entanglement \cite{EPR_35} in a non-local system \cite{BW_92, BBCJPW_93, BCWSW_12, AMP_12}. For the efficient implementation of many information processing tasks and different quantum algorithms, not only the presence of entanglement is necessary \cite{LP_01, V_13} but also it is the key property that helps a quantum structure to outperform it's classical counterpart \cite{JL_03}. Hence for all practical purposes, detecting entanglement in a given quantum system is one of the key domain of research in the corresponding literature. There are different detection criteria for entanglement in bipartite or multipartite systems \cite{GT_09} such as entanglement witnesses \cite{HHH_96, BCHHKLS_02}. Also to quantify the amount of the same in a given system, there are several well studied measures for bipartite systems, such as negativity \cite{VW_02}, concurrence \cite{HW_97, RBCHM_01} and many more \cite{HHH_09, PV_14}. While the field of detecting and measuring entanglement in bipartite system is a well explored area in quantum information theory \cite{ADH_08, GA_12, ZGZG_08, CHKST_16, GCGM_19}, moving to more parties and/or higher dimensions are still much untraveled in the literature. In most of the realisations of quantum algorithms the systems under consideration are multipartite and hence motivates the community to deal with the properties of entanglement in such scenarios. \\

The main mathematical ingredient of entanglement detection problem is positive but non completely positive maps. In some lower dimensions, such as for $2\times2$ and $2\times3$ dimensional states, these positive maps give rise to the necessary and sufficient condition for detecting entanglement and the corresponding map is partial transposition (PT) map \cite{HHH_96}. For higher dimensions this fails to give such criterion because there exist entangled states that are also positive under PT and this falls in the category of NP hard problems \cite{G_03}. The entanglement of such states are not distillable and hence they are called bound entangled states \cite{HHH_98}. These states are shown to be useful when the entanglement is unlocked \cite{HHH_99, SST_03}. Hence the area of study to detect bound entanglement is a fascinating and difficult area of research in quantum information theory \cite{LKPR_10, BSGMCRHB_10, SGSSK_18, BGMGCBM_21} and it is still far from exhaustive. It is already clear that a decomposable positive map can not detect a PPT entangled state [Ref]. So when a positive but not completely positive map detects PPT entanglement, then it is definitely an indecomposable map. As mentioned earlier, in higher dimensions PT fails to detect entanglement and hence people have introduced few other ways to detect entanglement in that regime, using computable cross norm or realignment criterion (CCNR criterion) \cite{R_05, CW_02}, range criterion \cite{BDMSST_99, BP_00} and others \cite{GHGE_07, GGHE_08, BGMGCBM_21}. Positive but non completely positive maps play the key role in such detection mechanisms. Here we try to put one step forward by trying to introduce a generic prescription to find such useful positive maps from a well defined structure. \\

From the literature of open quantum dynamics, we know that the most general quantum evolutions are represented by completely positive trace preserving (CPTP) maps, the generators of which (if exists) can be associated with Lindblad type super operators \citep{alicki,lindblad,gorini,breuer,samya2,Bhattacharya17,BBhattacharya21,Bhattacharya20,Maity20,Samya_eternal}. Study of Lindblad type dynamics is an integral part of research in open quantum systems. Interestingly we note that it is possible to construct a set of positive maps from a Lindblad operator by parameterizing the same. These positive maps are in turn useful to detect both bipartite and multipartite entanglement. In this paper, we first consider qubit systems. In this scenario, while detection of entanglement in bipartite quantum states is well studied in literature, moving to multipartite domain is much unexplored and non-trivial area. From the Lindblad structure we first construct an one parameter family of linear maps which acts on a qubit and helps detecting genuine multipartite entanglement (GME). We construct the corresponding witness operator and also give an universal measure for GME. The positive maps on this qubit systems is discussed thoroughly in Section \ref{sec2}. Next we move on to the problem of entanglement detection in higher dimension. Here also we construct suitable maps from the Lindblad structure. We check if the maps are decomposable or not as it has direct connection with detection of PPT entanglement. Finally from the general structure we construct various well known positive maps which are useful in detecting entanglement in higher dimensional systems. For this purpose we restrict ourselves to qutrit systems and the processes are described in Section \ref{sec3}. Finally in Section \ref{sec4}, we sum up the findings of the paper and try to open up some new directions to study in the literature of quantum information theory.

\section{Positive maps on $\mathcal{M}_2$}
\label{sec2}

Let $\mathcal{M}_2$ stands for the algebra of 2 by 2 complex matrices. Let us now define an one parameter family of linear maps $\Lambda$ on $\mathcal{M}_2$ as,
$\Lambda : \mathcal{M}_2 \longrightarrow \mathcal{M}_2 $ such that
\begin{eqnarray}
\Lambda (X) = && X + \gamma (\sigma_1 X \sigma_1 - \frac{1}{2} \sigma_1 \sigma_1 X - \frac{1}{2} X \sigma_1 \sigma_1) \nonumber\\ &-& \gamma (\sigma_2 X \sigma_2 - \frac{1}{2} \sigma_2 \sigma_2 X - \frac{1}{2} X \sigma_2 \sigma_2) \nonumber\\&+& \frac{1}{2} (\sigma_3 X \sigma_3 - \frac{1}{2} \sigma_3 \sigma_3 X - \frac{1}{2} X \sigma_3 \sigma_3) 
\label{lambda}
\end{eqnarray}
Here, $X \in \mathcal{M}_2$, $\gamma \in \mathbb{R}$, and
$\sigma_i$'s (for $i=1,2,3$) are the Pauli matrices.

The construction of the map is motivated from the structure of time independent Lindblad form. Further simplification leads the linear map to the given form as,
\begin{eqnarray}
\Lambda (X) = \begin{bmatrix}
r_{11}& 2 r_{21} \gamma\\
2 r_{12} \gamma & r_{22}
\end{bmatrix}  
\end{eqnarray}
for any $X= \begin{bmatrix}
r_{11}&r_{12}\\
r_{21}&r_{22}
\end{bmatrix} \in \mathcal{M}_2$ and $\gamma \in \mathbb{R}$. Note that, for $\gamma = \frac{1}{2}$ it reproduces the famous transposition map $\mathcal{T}$. Thus this family of map can be regarded as an one parameter generalisation of transposition map arising from Lindblad like structure. Hence it is interesting to search the range of the parameter $\gamma$ for which the map is positive and can potentially contribute to the literature of the problem of entanglement detection. Moreover we try to find whether it is completely positive (CP) for some some $\gamma$. Clearly for $\gamma = \frac{1}{2}$ the map is positive but not completely positive. \\

Let us consider a unit vector $\boldsymbol{\eta}$ in two dimensional Hilbert space. Applying the map $\Lambda$ on $\eta \eta^{*}$, where $\boldsymbol{\eta}^{*}$ is the adjoint of $\boldsymbol{\eta}$, we find that the family of maps $\Lambda$ is positive for $-1/2 \leq \gamma \leq 1/2$. Moreover, computing the Choi matrix of the corresponding family of maps we find that the maps are not completely positive for any $\gamma \in \mathbb{R}$. Therefore we have an one parameter family of positive but not completely positive (PNCP) maps $\Lambda$ for $-1/2 \leq \gamma \leq 1/2$ which are certainly useful for detecting entangled states.\\

\begin{prop}
     The minimum eigenvalue of the Choi state in $2\otimes2$ is the minimum eigenvalue for any state when subjected to $\Lambda$.
\end{prop}
\begin{proof}
Let us now try to find the minimum eigenvalue of the output matrix corresponding to any pure input state. Without loss of generality we can consider an input pure state in its Schmidt form. A pure state can be of two types: either product or entangled. For a product input state the eigenvalues of the output matrix is non-negative always as the map is positive. So to get the minimum eigenvalue we need to consider entangled states as an input. Let us consider
 $\ket {\psi} = c_1 \ket {00} + c_2 \ket{11}$, where $\vert c_1\vert^2 + \vert c_2\vert^2 =1 $. Now $\Lambda \otimes \mathcal{I}(\ket {\psi} \bra{\psi})$ produces the output matrix as
 \[
 \begin{bmatrix}
   \vert c_1\vert^2 &0&0&0\\
   0&0& 2c_2 c_{1}^{*}\gamma &0\\
   0& 2c_1 c_{2}^{*}\gamma&0&0\\
   0&0&0&\vert c_2\vert^2
 \end{bmatrix},
 \]
whose minimum eigenvalue is given by $\mathcal{E}_{min} = -2 \vert c_1\vert\vert c_2\vert \gamma$. Clearly, since $\vert c_1\vert^2 + \vert c_2\vert^2 =1 $, $\vert c_1\vert\vert c_2\vert \leq \frac{1}{2}$. Here the upper bound is achieved at $c_1 =c_2=\frac{1}{\sqrt{2}}$. Therefore the minimum eigenvalue of the output matrix is $-\gamma$ and it is obtained for two qubit maximally entangled state as input. Clearly for $\gamma= \frac{1}{2}$  we are getting back the previous result of transposition map. 
\end{proof}

It is well known that the transposition map plays an instrumental role in entanglement detection. Seminal results by Stormer, Woronowich and Horodecki's have established that transposition is the only PNCP map acting on $\mathcal{M}_2$ which is responsible for entanglement detection in two qubit or qubit-qutrit systems \cite{HHH_96}. In this paper we do not focus on entanglement detection on two qubit or qubit-qutrit systems. Rather we are interested in the effectiveness of the map in detecting multipartite entanglement. The domain of entanglement detection in this scenario becomes drastically non trivial as it involves the concept of genuine multipartite entanglement along with bi-separability as more than two parties are involved. In this paper, we concentrate on the simplest case in this genre where there are three qubits present.

\begin{center}
\textit{Multipartite entanglement}
\end{center}

\begin{figure*}[htp]
\centering
\fbox{
\subfigure[Eigenvalues when the map acts on $W$ state.]{\includegraphics[scale=0.6]{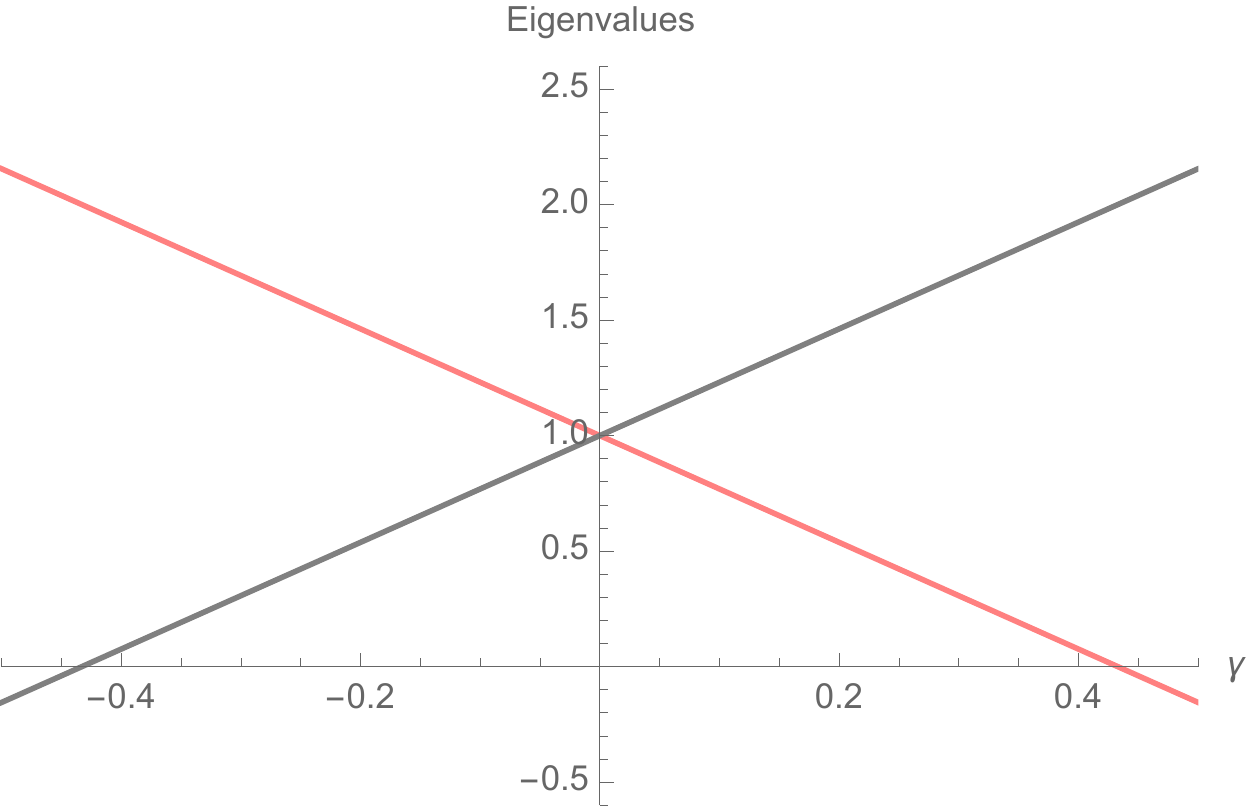}}
\qquad
\subfigure[Eigenvalues when the map acts on $GHZ$ state.]{\includegraphics[scale=0.6]{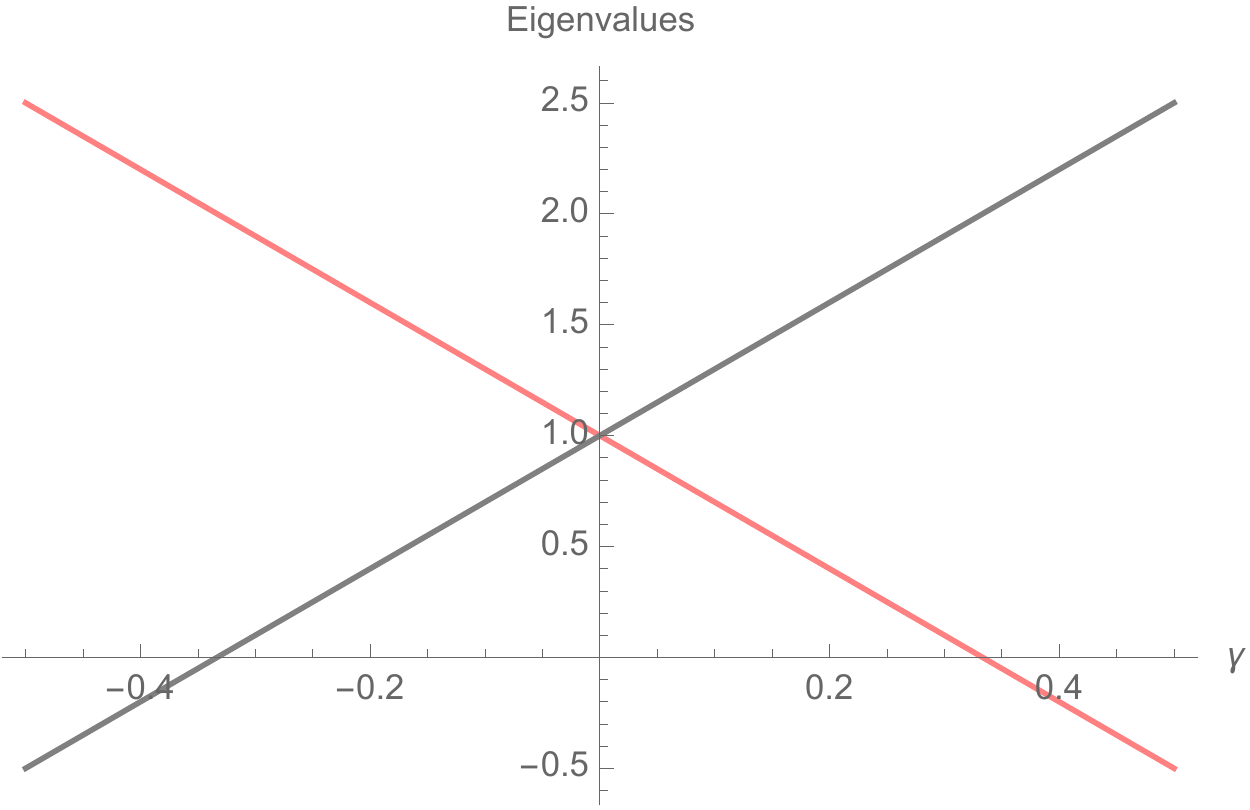}}}
\caption{\footnotesize{Variation of the eigenvalues which becomes negative with the map parameter $\gamma$ and hence detecting the corresponding states.}}
\label{N_E}
\end{figure*}

In this subsection, we concentrate on detecting genuinely multipartite entanglement by using positive maps. Let us consider a positive linear map $\tilde{\Lambda}$ corresponding to the positive map $\Lambda : \mathcal{M}_2 \longrightarrow \mathcal{M}_2 $ and the action of the map on a three-qubit system can be seen as, \\
\begin{eqnarray}
\tilde{\Lambda} [*] &&=(\Lambda_A \otimes \mathcal{I}_B \otimes \mathcal{I}_C+ \mathcal{I}_A \otimes \Lambda_B \otimes \mathcal{I}_C \nonumber\\
&+& \mathcal{I}_A \otimes \mathcal{I}_B \otimes\Lambda_C + c~Tr \mathcal{I})[*]
\label{lambda_t}
\end{eqnarray}
where A, B and C are the labels of the parties and $\mathcal{I}_X$ is the identity map on party X $\in  \lbrace A, B,C \rbrace$ and $c$ is a constant. We can call this map as the lifting of the positive map $\Lambda$. It was shown in paper \cite{HS_14}, that the similar lifting of the transposition map $\mathcal{T}$ can detect three qubit $W$ states. Our proposed map is a one parameter generalization of the Transposition map arising from the Lindblad structure. We have found that $\tilde{\Lambda} (\vert W \rangle \langle W \vert) \ngeq 0$ for the values of $\gamma \in [-\frac{1}{2}, -\frac{\sqrt{3}}{4}) \bigcup ( \frac{\sqrt{3}}{4}, \frac{1}{2}]$. Clearly $\gamma = \frac{1}{2}$ corresponds to the the transposition map.\\

Now we consider the case when the $W$ state is mixed with some white noise. Hence the state can be written as, $\tilde{W}= p \vert W \rangle \langle W \vert + (1-p) \frac{\mathbb{I}}{8} $. It was shown in \cite{HS_14, CHLM_17, VMB_22} that for $\gamma = \frac{1}{2}$ genuine entanglement can be detected for all $p$ strictly greater than 0.73. Our finding is that the same is also true for $\gamma = -\frac{1}{2}$. Moreover $\gamma = \frac{1}{2}$  and  $\gamma = -\frac{1}{2}$ are the best choice for the class of noisy $W$ state as for $\gamma \in (-\frac{2}{5}, \frac{2}{5})$ no genuine entanglement can be detected by $\tilde{\Lambda}$. Beyond the above mentioned region, $\gamma = \frac{1}{2}$  and  $\gamma = -\frac{1}{2}$ are the only values of the map parameter for which the largest class of genuinely entangled noisy $W$ states can be detected by this maps. In a similar way GHZ state can also be detected by using the map $\tilde{\Lambda}$ followed by a rotation by $\sigma_x$ operator. The result has been graphically represented in Figure \ref{N_E}.\\

\par
\paragraph*{ \textbf{Construction of witness for genuine multipartite entanglement (GME):}}

Further, we can also construct a witness for detection of GME in the multipartite scenario from the above class of transposition maps. A witness \(\mathcal{W}\) is a Hermitian operator such that \(Tr[\mathcal{W}\sigma_{2-sep}]\geq0\) for all separable and bi-separable states $\sigma_{2-sep}$ and \(Tr[\mathcal{W}\sigma_{GME}]\ngeq0\) for at least one genuinely entangled state. We know that entanglement witnesses are directly measurable quantities and hence a very useful tool for detecting entanglement in experiments.

\begin{prop}
\(\mathcal{W}_{GME} = \tilde{\Lambda}(\ket{W}\bra{W})\) is a witness detecting GME where \(\ket{W} = \dfrac{1}{\sqrt{3}}(\ket{001}+\ket{010}+\ket{100})\).
\end{prop}

\begin{proof}
Let us consider the map \(\tilde{\Lambda}\) introduced in Eq. (\ref{lambda_t}) obtained by lifting the map \(\Lambda\) action of which has been shown in Eq. (\ref{lambda}) and can be simplified as, 
\begin{eqnarray}
\Lambda(\rho) = \rho + \sum_{i} \gamma_i (\sigma_i \rho \sigma_i - \dfrac{1}{2}(\sigma_i \sigma_i \rho + \rho \sigma_i \sigma_i))
\label{lambda_alt}
\end{eqnarray}
Expanding as given in Eq. (\ref{lambda_alt}), it can be easily seen that,  $Tr[\Lambda(\rho)\ket{W}\bra{W}] = Tr[\Lambda(\ket{W}\bra{W})\rho]$. 
Hence we have,
\begin{align}Tr[\tilde{\Lambda}(\sigma_{2-sep})\ket{W}\bra{W}] = Tr[\tilde{\Lambda}(\ket{W}\bra{W})\sigma_{2-sep}] 
\end{align}
We know from [Ref] that,
\[\tilde{\Lambda}(\sigma_{2-sep}) \geq 0\]
\[\Rightarrow Tr[\tilde{\Lambda}(\sigma_{2-sep})\ket{W}\bra{W}] \geq 0 \]
Hence we get,
\begin{align}
    Tr[\tilde{\Lambda}(\ket{W}\bra{W})\sigma_{2-sep}] \geq 0
    \label{wit}
\end{align}

The minimum eigenvalue for the map \(\Lambda\) is \(\gamma\) i.e \(\mathcal{E}_{min}(\Lambda) = \gamma\). Hence, the constant c in Eq. (\ref{lambda_t}) is \(2\gamma\). We can easily see that \(\tilde{\Lambda}(\sigma_{2-sep})\geq0\) when c is atleast \(2\gamma\) for the tripaprtite scenario. Also, we found that \(\tilde{\Lambda}(\ket{W}\bra{W}) \ngeq0\) for the above range of \(\gamma\). Finally following Eq. (\ref{wit}), \(\mathcal{W}_{GME} = \tilde{\Lambda}(\ket{W}\bra{W})\) acts as GME witness.
\end{proof}

\begin{remark}
Since the proof works for any general \(\tilde{\Lambda}\) constructed from the Lindbladian structure, all the obtained maps in this paper have corresponding witnesses that can obtained in a similar way.
\end{remark}

\paragraph*{\textbf{Negativity for GME :}}
In bipartite systems a well known measure for entanglement is Negativity and it is defined in terms of the negative eigenvalue of the partially transposed density matrix. Here, we move one step forward by introducing a similar measure for genuine multipartite entanglement (GME). Let us consider the lifted Transposition map $\tilde{\mathcal{T}}$. One should note that $\tilde{\mathcal{T}}$ is not trace preserving. Therefore let us consider the normalized map $\hat{\tilde{\mathcal{T}}}= \frac{1}{\text{N}}\tilde{\mathcal{T}} $, where N is the normalizing factor. It is defined as follows,
\begin{eqnarray}
\mathcal{N}_{GME}(\rho) = \dfrac{\|\hat{\tilde{\mathcal{T}}}[\rho]\|_{1}-1}{\mathcal{K}},
\label{GME_measure}
\end{eqnarray}
where $\mathcal{K}$ is a factor of normalization and to be chosen suitably.


\begin{prop}
$\mathcal{N}_{GME}$ introduced in Eq. (\ref{GME_measure}) is a valid entanglement measure. 
\end{prop}
\begin{proof}
To prove the above proposition, we must show that,
\begin{enumerate}
\item 
\(\mathcal{N}_{GME}(\rho_{2-sep}) = 0\)
\item
\(\mathcal{N}_{GME}\) is monotone under LOCC.
\item
\(\mathcal{N}_{GME}\) is convex.
\end{enumerate}
Here $\rho_{2-sep}$ denotes any bi-separable state and LOCC stands for local operation and classical communication.\\

For any Hermitian operator A, the trace norm i.e \(\|A\|_1 = Tr\sqrt{A^\dagger A}\) is the sum of the absolute values of the eigenvalues of A. 
Here we have, \(\hat{\tilde{\mathcal{T}}}[\rho_{2-sep}] \geq 0\) as seen in previous section and also being trace preserving which implies \(\|\hat{\tilde{\mathcal{T}}}[\rho_{2-sep}]\|_{1} = 1\). Hence , $\mathcal{N}_{GME}(\rho_{2-sep}) = 0$. \\

Now, we show that \(\mathcal{N}_{GME}\) is a monotone under LOCC. In a LOCC protocol one of the constituent parties perform some local measurement and broadcast the result to other parties. Depending upon the outcome of one party the subsequent parties choose their respective local operations. Without loss of generality lets say Alice starts the protocol. If the initial state is $\rho$, then after the measurement of Alice the final unnormalised state becomes
\[\mathcal{M}_i(\rho) = (M_i \otimes \mathcal{I}_B \otimes \mathcal{I}_C) \rho (M_{i}^\dagger \otimes \mathcal{I}_B \otimes \mathcal{I}_C)\]
where \[\sum_i M_i^\dagger M_i \leq \mathcal{I}_A.\]


\par
We note that,
\begin{align}
    \mathcal{M}_i(\mathcal{T}\otimes \mathcal{I}_B \otimes \mathcal{I}_C[\rho]) = \mathcal{T}\otimes \mathcal{I}_B \otimes \mathcal{I}_C(\mathcal{M}_i(\rho))
\end{align}

Due to symmetry, Bob and Charlie can perform the similar operation. Moreover this invariance is true even if they do nothing. In this scenario one round of LOCC is considered whenever Alice, Bob and Charlie complete their operation once. Hence following Vidal's approach \cite{VW_02} we can establish the monotonicity of $\mathcal{N}_{GME}$ under LOCC. Hence the claim.
\end{proof}

As an example of the usefulness of the measure, we calculate it for the state $\ket{W}$ mixed with a maximally mixed state by a varying noise parameter $p$ given by 
\[\rho_W=p\ket{W}\bra{W}+(1-p)\frac{\mathbb{I}}{8}.\]
Note that the measure takes a non-zero value from $p=0.9$ to $p=1$ as there does not exist any negative eigenvalues of the density matrix after the action of the map other than the mentioned region. We show the variation of the same in the Figure \ref{W_p_GME}.

\begin{figure}[htp]
\fbox{\includegraphics[scale=0.5]{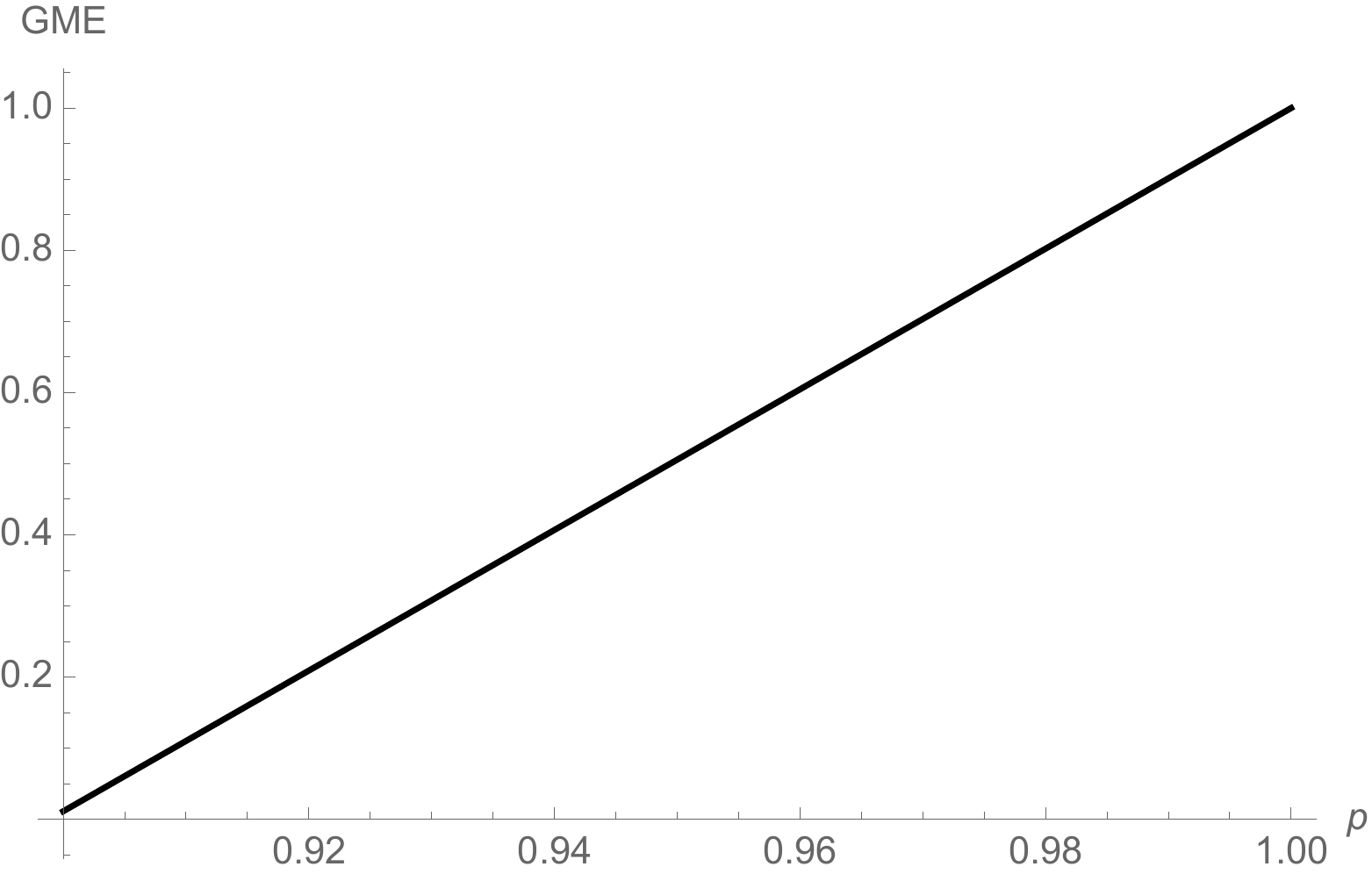}}
\caption{\footnotesize{Variation of the measure of GME with the white noise parameter introduced in $\rho_W$. For p=1, it gives back $W$ state where the GME takes the maximum value.}}
\label{W_p_GME}
\end{figure}

\section{Positive maps on $\mathcal{M}_3$}
\label{sec3}

In this section, we consider positive maps in higher dimensions, arising from the Lindblad structure. Other than generalising our method, the reason behind considering higher dimensional maps is to explore the possibilities of both decomposable and indecomposable maps arising from the structure of Lindbladians. One of the maps we will look into, can be called a Choi-like map which turns out to be decomposable and not stronger than the famous transposition map. On the other hand we will also derive the indecomposable Choi map useful for detecting bound entanglement by the proposed method. \\

Let $\{G_{i}\}$ be the set of Gell Mann matrices with \(i = 1,2,3,..,9\) and $G_0=\mathcal{I}_3$ is the 3-dimensional identity matrix. In the similar fashion with the previous section the algebra of \(3 \times 3\) complex matrices is denoted by $\mathcal{M}_3$. We define a linear map, $\Phi : \mathcal{M}_{3} \rightarrow \mathcal{M}_{3}$ such that,
\begin{eqnarray}
\Phi[\rho] = (\mathcal{I} + L)[\rho],~~\forall \rho \in \mathcal{M}_{3}
\end{eqnarray}
where $L : \mathcal{M}_{3} \rightarrow \mathcal{M}_{3}$ is another linear map given by,
\begin{eqnarray}
L[\rho] = \sum_{i} \gamma_{i}(G_{i}\rho G_{i} - \frac{1}{2}(G_{i}G_{i}\rho + \rho G_{i}G_{i})).
\label{LindbSum}
\end{eqnarray}

From this structure, we are now going to construct various well known positive maps in the following.

\subsection{Transposition map}

Note that for some specific values for the Lindblad coefficients, we get back the transposition map. By putting \(\gamma_{2}=\gamma_{4}=\gamma_{5}=\gamma_{7}=1/2 , \gamma_{3}=\gamma_{6}=\gamma_{8}=-1/2, \gamma_{9} = 1/6 \), we have the transposition map. Now if we generalize the map by considering \(\gamma_{2}=\gamma_{4}=\gamma_{5}=\gamma_{7}= \alpha, \gamma_{3}=\gamma_{6}=\gamma_{8}=-\alpha\), and \(\gamma_{9} = \alpha/3 \), we find an one parameter family of maps \(\Phi_{\alpha}\) which is positive for the range \(0 \leq \alpha \leq 1/2\).\\

First let us write the action of the map explicitly.  $\Phi_{\alpha} : \mathcal{M}_3 \rightarrow \mathcal{M}_3 $ is given by,\\
\begin{eqnarray}
&&\Phi_{\alpha}(X) = \nonumber\\ 
&& \begin{scriptsize} \begin{bmatrix}
x11& x_{12}-2 x_{12}\alpha+ 2x_{21}\alpha &~~ x_{13}-2x_{13}\alpha+2x_{31}\alpha\\
x_{21}+2x_{12}\alpha -2x_{21}\alpha & x_{22}& x_{23}-2x_{23}\alpha+2 x_{32}\alpha \\
x_{31}+2x_{13}\alpha -2x_{31}\alpha &~~ x_{32}+2x_{23}\alpha-2x_{32}\alpha & x_{33}
\end{bmatrix}
\end{scriptsize} \nonumber\\
\end{eqnarray}
for any \[X= \begin{bmatrix}
    x_{11}&x_{12}&x_{13}\\
    x_{21}&x_{22}&x_{23}\\
    x_{31}&x_{32}&x_{33}
\end{bmatrix} \in \mathcal{M}_3\]
 
To check the positivity of the map, it suffices to check positivity for pure states. Hence, for a general pure state \(\rho = \ket{\psi}\bra{\psi} \) where \(\ket{\psi} = [\psi_{1},\psi_{2},\psi_{3}]^{T}\), the matrix \(\Phi_{\alpha}[\rho]\) has to be a positive semi-definite. Thus, the determinants of all the principal minors of the matrix must be greater than or equal to zero \citep{pos1}. We have the principal minors as,
$2\alpha(2\alpha -1)(\psi_{2} \bar{\psi_{3}} - \psi_3\bar{\psi_2})^2$, 
$2\alpha(2\alpha -1)(\psi_{1} \bar{\psi_{3}} - \psi_3\bar{\psi_1})^2$,
$2\alpha(2\alpha -1)(\psi_{1} \bar{\psi_{2}} - \psi_2\bar{\psi_{1}})^2$.
For the map to be positive, we must have the abouve mentioned quantities to be greater than or equal to zero which is obtained for \(0 \leq \alpha \leq 1/2\). Therefore, the map is positive for this range of \(\alpha\).\\

Note that for the map to be able to detect entanglement we need the map to be positive but not completely positive. Hence we check at what range of $\alpha$, the map is not completely positive. We do this by checking if the corresponding Choi matrix is a positive semi definite matrix or not. The Choi matrix corresponding to the map is given by,
$\Psi_C= (\mathcal{I} \otimes \Phi_{\alpha})[\rho_{ent}]$ with, 
$\rho_{ent} = \ket{\psi_{ent}}\bra{\psi_{ent}}$ where, $\ket{\psi_{ent}} = \dfrac{1}{\sqrt{d}} \sum_{i=0}^{d-1}\ket{ii} = \dfrac{1}{\sqrt{3}}(\ket{00}+\ket{11}+\ket{22})$. The distinct eigenvalues for the Choi matrix for the map \(\Phi_{\alpha}\) are \{\(1 - \dfrac{4\alpha}{3},-\dfrac{2\alpha}{3},\dfrac{2\alpha}{3}\)\}. Hence, the map is not complete positive for \(\alpha \geq 0\) and finally positive but not completely positive for \(0 \leq \alpha \leq 1/2\).\\

Next We show that the above class of maps is decomposable and cannot detect any PPT entangled states. We note that the action of the map $\Phi_{\alpha}$ can be decomposed as the following:\\
\begin{eqnarray}
\Phi_{\alpha} (X)= (1-2 \alpha) \begin{bmatrix}
    x_{11}&x_{12}&x_{13}\\
    x_{21}&x_{22}&x_{23}\\
    x_{31}&x_{32}&x_{33}
\end{bmatrix} + 2 \alpha \begin{bmatrix}
    x_{11}&x_{21}&x_{31}\\
    x_{12}&x_{22}&x_{32}\\
    x_{13}&x_{23}&x_{33}
\end{bmatrix} \nonumber\\
\end{eqnarray}
for any $X \in \mathcal{M}_3$. Hence, the map can be written as,
\begin{eqnarray}
\Phi_{\alpha} = (1-2 \alpha) \mathcal{I}+ 2 \alpha \mathcal{T}\circ \mathcal{I}
\end{eqnarray}
where $\mathcal{I}$ and $\mathcal{T}$ are the identity map and transposition map respectively. Identity map being a completely positive, $\Phi_{\alpha}$ becomes decomposable map in the range \(0 \leq \alpha \leq 1/2\) as it can be written as sum of a completely positive map and a completely co-positive map. 
\\

\subsection{Decomposable Choi map}
Note that for \(\gamma_{2}=1/4 , \gamma_{3}=1/4, \gamma_{4}=1/4, \gamma_{5}=1/4, \gamma_{6}=1/4, \gamma_{7}=1/4, \gamma_{8}=1/4, \gamma_{9}=1/12 \), we have one of the Choi maps i.e.
\begin{eqnarray}
\Phi^{2}[\rho] = 
\begin{bmatrix}
    \dfrac{\rho_{22}+\rho_{33}}{2}       & -\dfrac{\rho_{12}}{2} & -\dfrac{\rho_{13}}{2} \\
    -\dfrac{\rho_{21}}{2}       & \dfrac{\rho_{11}+\rho_{33}}{2} & -\dfrac{\rho_{23}}{2} \\
    -\dfrac{\rho_{31}}{2}       & -\dfrac{\rho_{32}}{2} & \dfrac{\rho_{11}+\rho_{22}}{2}
\end{bmatrix}
\end{eqnarray}
Similarly, we parameterize it with \(\alpha\) i.e\ \(\gamma_{2}=\gamma_{3}=\gamma_{4}=\gamma_{5}=\gamma_{6}=\gamma_{7}=\gamma_{8}=\alpha\) and \(\gamma_{9}=\alpha/3\), where we find the map \(\Phi^{2}_{\alpha}\) is positive for the range \(0 \leq \alpha \leq 1/4\). On the other hand the map is not complete positive for \( \alpha > 3/16\). Hence, the map is Positive but not completely positive for \(3/16 < \alpha \leq 1/4\).
From Theorem 3.4 of the Generalized Choi maps in three-dimensional matrix algebra of \cite{CHO_33}, we get a necessary and sufficient condition for a map on qutrit system to be decomposable. And the Choi like map \(\Phi^{2}_{\alpha}\) we obtain is proved to be decomposable family of maps. We cannot obtain the original Choi map from this structure for any values of the \(\gamma_{i}'s\).
\\

\subsection{Indecomposable Choi map}
Though we cannot obtain the well known indecomposable Choi map from (\ref{LindbSum}), we can however reconstruct the original Choi map by taking a combination sum of Lindblad-type sums. Consider the map \(\Phi^{C}\) given below
\begin{eqnarray}
\Phi^{C}(\rho) = (\sum_{i=1}^{3} A_{i}\rho A_{i}^{\dagger} -\frac{1}{2}\{A_{i}A_{i}^{\dagger},\rho\}) \nonumber
\\- (\sum_{j=1}^{3} B_{j}\rho B_{j}^{\dagger} -\frac{1}{2}\{B_{j}B_{j}^{\dagger},\rho\}) \nonumber
\\+ (\sum_{k=1}^{3} C_{k}\rho C_{k}^{\dagger} -\frac{1}{2}\{C_{k}C_{k}^{\dagger},\rho\}) \nonumber
\end{eqnarray}
where \(A_{1} = \ket{1}\bra{2}, A_{2} = \ket{2}\bra{3}, A_{3} = \ket{3}\bra{1}\) and \(B_{1} = \ket{1}\bra{1}, B_{2} = \ket{2}\bra{2}, B_{3} = \ket{3}\bra{3}\) and \(C_{1} = \ket{1}\bra{1}-\ket{2}\bra{2}, C_{2} = \ket{2}\bra{2}-\ket{3}\bra{3}, C_{3} = \ket{3}\bra{3} - \ket{1}\bra{1}\) respectively. We can see that the following map
\[\Phi_F^C(\rho)=\left[\mathcal{I}+\frac{1}{2}\Phi^{C}\right](\rho),\]
is nothing but the well known Choi map,
\begin{eqnarray}
\Phi_F^{C}[\rho] = 
\begin{bmatrix}
    \rho_{11}+\rho_{22}       & -\rho_{12} & -\rho_{13} \\
    -\rho_{21}      & \rho_{22}+\rho_{33} & -\rho_{23} \\
    -\rho_{31}      & -\rho_{32} & \rho_{33}+\rho_{11}
\end{bmatrix}.
\end{eqnarray}
We would want to parameterize the above structure to find a range of Choi-like indecomposable maps. We can consider the following one parameter family of maps,
\begin{eqnarray}
\Phi^{C}_{\beta}[\rho] = S_1 - \beta(S_2 - S_3)
\label{LindbSums}
\end{eqnarray}
where \(\beta\) is the single parameter and at \(\beta = 1\), we get back the Choi map \(\Phi^{C}\). 
\begin{eqnarray}
S_1 = \sum_{i=1}^{3} A_{i}\rho A_{i}^{\dagger} -\frac{1}{2}\{A_{i}A_{i}^{\dagger},\rho\} \nonumber
\\S_2 = \sum_{j=1}^{3} B_{j}\rho B_{j}^{\dagger} -\frac{1}{2}\{B_{j}B_{j}^{\dagger},\rho\} \nonumber
\\S_3 = \sum_{k=1}^{3} C_{k}\rho C_{k}^{\dagger} -\frac{1}{2}\{C_{k}C_{k}^{\dagger},\rho\} \nonumber
\end{eqnarray}


Note that, the map \(\Phi_{F, \beta}^C\) (which is the $\beta$ parameterized linear map corresponding to the positive map $\Phi_F^{C}$) is a positive map for the range $0\leq \beta \leq 1$ and it is completely positive for $\beta \leq 3/4$.

\section{Conclusion}
\label{sec4}

In the literature of quantum information theory there exists an involved connection between entanglement detection and positive but not completely positive map. Also as entanglement is a basic and vastly used quantum resource in different information processing tasks, detecting it in an unknown system is one of the most crucial jobs in any experimental scenario. In this work, we give a generalized prescription to construct such positive maps from a given structure. Here we have investigated the structure of Lindblad superoperators for the case of qubits and qutrits and constructed various well known positive maps from them. We have found that in different parameter regions, we can reconstruct the transposition map and different variants of Choi map from the structure of the Lindbladians. We have further investigated a process of GME detection by exploiting positive maps, for the case of three qubits scenario. We have also constructed a  class of linear GME witness operators. Moreover, we have also been able to propose a negativity measure for GME, based on a modified map constructed out of transposition. This gives a non zero value only for GME states, in turn  distinguishing the bi-separable states from them. As an example we find the quantification of the measure for the noisy $W$ state. Our work can be directly generalized to higher dimensions and hence gives rise to a novel method to both detect and measure bi-partite and multi-partite entanglement.

\section{Acknowledgement}

BB acknowledges Ritabrata Sengupta for various discussions on the theory of positive maps.

%

\end{document}